\documentclass{amsart}

\usepackage{amssymb}
\usepackage{hyperref}
\usepackage[all]{xy}
\usepackage{verbatim}
\usepackage{ifthen}
\usepackage{xargs}
\usepackage{bussproofs}
\usepackage{etex}

\hypersetup{colorlinks=true,linkcolor=blue}

\newcommand{\axref}[1]{(\hyperref[ax:#1]{#1})}

\newcommand{\newref}[4][]{
\ifthenelse{\equal{#1}{}}{\newtheorem{h#2}[hthm]{#4}}{\newtheorem{h#2}{#4}[#1]}
\expandafter\newcommand\csname r#2\endcsname[1]{#3~\ref{#2:##1}}
\expandafter\newcommand\csname R#2\endcsname[1]{#4~\ref{#2:##1}}
\newenvironmentx{#2}[2][1=,2=]{
\ifthenelse{\equal{##2}{}}{\begin{h#2}}{\begin{h#2}[##2]}
\ifthenelse{\equal{##1}{}}{}{\label{#2:##1}}
}{\end{h#2}}
}

\newref[section]{thm}{Theorem}{Theorem}
\newref{lem}{Lemma}{Lemma}
\newref{prop}{Proposition}{Proposition}
\newref{cor}{Corollary}{Corollary}

\theoremstyle{definition}
\newref{defn}{Definition}{Definition}
\newref{example}{Example}{Example}

\theoremstyle{remark}
\newref{remark}{Remark}{Remark}

\newcommand{\fs}[1]{\mathrm{#1}}
\newcommand{\lcon}{\fs{left}}
\newcommand{\rcon}{\fs{right}}
\newcommand{\Path}{\fs{Path}}
\newcommand{\pcon}{\fs{path}}
\newcommand{\I}{\fs{I}}
\newcommand{\coe}{\fs{coe}}
\newcommand{\iso}{\fs{iso}}
\newcommand{\id}{\fs{id}}
\newcommand{\U}{\mathcal{U}}
\newcommand{\hU}{\widehat{\mathcal{U}}}
\newcommand{\Eq}{\fs{Eq}}
\newcommand{\qEq}{\fs{qEq}}
\newcommand{\ev}{\fs{ev}}

\numberwithin{figure}{section}

\newcommand{\pb}[1][dr]{\save*!/#1-1.2pc/#1:(-1,1)@^{|-}\restore}

\begin{document}

\title{Models of Homotopy Type Theory with an Interval Type}

\author{Valery Isaev}

\begin{abstract}
In this short note, we construct a class of models of an extension of homotopy type theory, which we call homotopy type theory with an interval type.
\end{abstract}

\maketitle

\section{Introduction}

Homotopy type theory with an interval type (HoTT-I) is a simple extension of the ordinary homotopy type theory, which is implemented in Arend proof assistant.
Instead of identity types, HoTT-I has the interval type and path types defined as certain functions from it.
The univalence axiom is defined in HoTT-I in such a way that the composite $\Sigma_{(f : A \to B)} \mathrm{isEquiv}(f) \to A =_\mathcal{U} B \to (A \to B)$ is definitionally equal to the first projection.
HoTT-I can also be extended with higher inductive types which satisfy computational $\beta$-rules even for higher constructors, but we will not discuss them in this note.
The theory is similar to cubical type theory \cite{cubical-tt}, but it does not satisfy the canonicity property.

In this note, we will show that (a basic version) HoTT-I can be interpreted in any right proper Cartesian model category in which cofibrations are precisely monomorphisms and which is locally Cartesian closed as a category assuming it has enough univalent universes.
A model category is Cartesian if, for every pair of cofibrations $f : A \to C$ and $g : B \to D$, their pushout-product $f \square g : A \times D \amalg_{A \times C} B \times C \to B \times D$ is a cofibration and is a trivial cofibration whenever one of the maps $f$ and $g$ is.
We let $\mathcal{M}$ be a fixed model category satisfying these properties, which will be used throughout this note.

We will use the universe construction \cite{kap-lum-voe} to solve the coherence issues.
Thus, we will assume that there is a fixed fibration $\pi : \hU \to \U$ in $\mathcal{M}$, which satisfies the univalence axiom and classifies a class of fibrations closed under all necessary constructions.
A dependent type $\Gamma \vdash A$ will be interpreted as a map $v_A : \Gamma \to \U$.
Terms of type $A$ are interpreted as sections of $\pi$ over $v_A$.
If $v_A$ factors as $\Gamma \xrightarrow{f} V_A \xrightarrow{g} \U$, then terms of type $A$ can be equivalently described as sections of $g^*(\pi)$ over $f$.
We will use both of these definitions.

\section{Interval type}
\label{sec:interval}

The interval type is just the unit type with two constructors: $\lcon$ and $\rcon$.
The eliminator for it is the same as for the unit type and will be denoted by $\coe$:

\medskip
\begin{center}
\AxiomC{$\Gamma \vdash$}
\UnaryInfC{$\Gamma \vdash \I$}
\DisplayProof
\quad
\AxiomC{$\Gamma \vdash$}
\UnaryInfC{$\Gamma \vdash \lcon : \I$}
\DisplayProof
\quad
\AxiomC{$\Gamma \vdash$}
\UnaryInfC{$\Gamma \vdash \rcon : \I$}
\DisplayProof
\end{center}

\medskip
\begin{center}
\AxiomC{$\Gamma, x : \I \vdash A$}
\AxiomC{$\Gamma \vdash a : A[x := \lcon]$}
\AxiomC{$\Gamma \vdash i : \I$}
\TrinaryInfC{$\Gamma \vdash \coe(x. A, a, i) : A[x := i]$}
\DisplayProof
\end{center}

\[ \coe(x. A, a, \lcon) \equiv a \]

The interval type can be interpreted as the terminal object, but this will give us a model satisfying the K axiom.
To get a homotopic model, we will interpret $\I$ as an interval object.
We can take $\I$ to be any contractible fibrant object with a cofibration from $1 \amalg 1$, but we will actually assume that we have a cofibration from $4 = 1 \amalg 1 \amalg 1 \amalg 1$; this will be useful later.
Thus, let $4 \to \I \to 1$ be a factorization of the map $4 \to 1$ into a cofibration followed by a trivial fibration.
Constructors $\lcon$ and $\rcon$ are interpreted as maps $1 \to \I$ corresponding to the first and second coprojections of the map $4 \to \I$.
If $\Gamma$ is any object of $\mathcal{M}$, we will also denote by $\lcon$ and $\rcon$ the maps $\langle \id, \lcon \circ !_\Gamma \rangle : \Gamma \to \Gamma \times \I$ and $\langle \id, \rcon \circ !_\Gamma \rangle : \Gamma \to \Gamma \times \I$, respectively.

Let us describe the interpretation of $\I$ and $\coe$.
We need to assume that $\I$ is classified by a map $\chi_\I : 1 \to \U$.
Then we define $v_\I : \Gamma \to \U$ as $\chi_\I \circ !_\Gamma$.
To describe $\coe(A,a,i)$, consider the following pullback:
\[ \xymatrix{ T \ar[r]^e \ar[d]_d \pb                   & \hU \ar[d]^\pi \\
              \mathcal{U}^\I \ar[r]_{\ev \circ \lcon}   & \U
            } \]
Then $v_A : \Gamma \times \I \to \U$ and $a : \Gamma \to \hU$ determine a map $c : \Gamma \to T$.
The interpretation of $\coe(A,a,i)$ is $\Gamma \xrightarrow{\langle c, i \rangle} T \times \I \xrightarrow{s} \hU$, where $s$ is a lift in the following diagram:
\[ \xymatrix{ T           \ar[rr]^e \ar[d]_\lcon                  &                             & \hU \ar[d]^\pi \\
              T \times \I \ar[r]_-{d \times \id} \ar@{-->}[urr]^s & \U^I \times I \ar[r]_-\ev   & \U
            }\]
If $i = \lcon$, then $\langle c, i \rangle$ factors through $e : T \to \hU$, which implies that $\coe$ satisfies the required computational rule.

We will add more computational rules for $\coe$ later.
Thus, we will need to modify the interpretation of $\coe$ to support them.
To do this, we will use the following general construction.
Let $C$ be an object of $\mathcal{M}$, let $f : C \to T$ be a cofibration, and let $g : C \times \I \to \hU$ be a map such that the following diagram commutes:
\[ \xymatrix{ C \ar[rr]^f \ar[d]_\lcon                    &                             & T \ar[d]^e \\
              C \times \I \ar[rr]^g \ar[d]_{f \times \id} &                             & \hU \ar[d]^\pi \\
              T \times \I \ar[r]_-{d \times \id}          & \U^\I \times \I \ar[r]_-\ev & \U
            } \]
Now, consider the following diagram:
\[ \xymatrix{ C \times \I \amalg_C T \ar[rr]^-{[g,e]} \ar[d]_{f \square \lcon} &                                & \hU \ar[d]^\pi \\
              T \times \I \ar[r]_-{d \times \id} \ar@{-->}[urr]^{s'}           & \U^\I \times \I \ar[r]_-\ev    & \U
            } \]
Conditions on $f$ and $g$ guarantee that the map $[g,e]$ is well-defined and that the diagram above consisting of solid arrows commutes.
Since $f$ is a cofibration and $\lcon$ is a trivial cofibration, $f \square \lcon$ is also a trivial cofibration.
Thus, we have a lift $s'$ in this diagram.
We can interpret $\coe(A,a,i)$ as $\Gamma \xrightarrow{\langle c, i \rangle} T \times \I \xrightarrow{s'} \hU$.
If $c : \Gamma \to T$ factors as $\Gamma \xrightarrow{c'} C \xrightarrow{f} T$, then $\coe(A,a,i)$ equals to $\Gamma \xrightarrow{\langle c', i \rangle} C \times \I \xrightarrow{g} \hU$, which will give us required additional computational rules.

Let $\{ (f_j : C_j \to T, g_j : C_j \times \I \to T) \}_{j \in \{1,2\}}$ be two pairs of maps satisfying the conditions given above.
Suppose that $\mathcal{M}$ is a topos.
Then we can define the union $f : C \to T$ of subobjects $f_1$ and $f_2$ as $C_1 \amalg_{C_0} C_2 \to T$, where $C_0$ is the intersection of $f_1$ and $f_2$.
Since $- \times \I$ commutes with colimits, $(C_1 \amalg_{C_0} C_2) \times \I$ is the pushout $C_1 \times \I \amalg_{C_0 \times \I} C_2 \times \I$.
Thus, we can define the map $g : (C_1 \amalg_{C_0} C_2) \times \I \to T$ determined by $g_1$ and $g_2$ if the following square commutes:
\[ \xymatrix{ C_0 \times \I \ar[r] \ar[d] & C_2 \times \I \ar[d]^{g_2} \\
              C_1 \times \I \ar[r]_-{g_1} & \hU
            } \]

By the universal property of pushouts, $f$ and $g$ satisfy the required conditions.
Thus, we obtained an interpretation of $\coe$ which satisfies computational rules corresponding to both $(f_1,g_1)$ and $(f_2,g_2)$.
That is, if $\mathcal{M}$ is a topos, we can combine two additional computational rules for $\coe$ as long as $g_1$ and $g_2$ corresponding to these rules satisfy the condition given above.
More generally, if we have a finite set of additional computational rules for $\coe$, then we just need to check that this condition holds pairwise.

\begin{remark}
Informally, the condition on $g_1$ and $g_2$ simply means that the right hand sides of the corresponding computational rules agree on the intersection of the left hand sides.
\end{remark}

\section{Path types}

Identity types are replaced with path types in HoTT-I:

\medskip
\begin{center}
\AxiomC{$\Gamma, x : \I \vdash A$}
\AxiomC{$\Gamma \vdash a : A[x := \lcon]$}
\AxiomC{$\Gamma \vdash a' : A[x := \rcon]$}
\TrinaryInfC{$\Gamma \vdash \Path(x. A, a, a')$}
\DisplayProof
\end{center}

\medskip
\begin{center}
\AxiomC{$\Gamma, x : I \vdash a : A$}
\UnaryInfC{$\Gamma \vdash \pcon(x. a) : \Path(x. A, a[x := \lcon], a[x := \rcon])$}
\DisplayProof
\end{center}

\medskip
\begin{center}
\AxiomC{$\Gamma \vdash p : \Path(x. A, a, a')$}
\AxiomC{$\Gamma \vdash i : \I$}
\BinaryInfC{$\Gamma \vdash p\ @_{a,a'}\ i : A[x := i]$}
\DisplayProof
\end{center}

\begin{align*}
& \pcon(x. t)\ @_{a,a'}\ i \equiv t[x := i] \\
& \pcon(x. p\ @\ x) \equiv p \text{ if } x \notin \fs{FV}(p) \\
& p\ @_{a,a'}\ \lcon \equiv a \\
& p\ @_{a,a'}\ \rcon \equiv a'
\end{align*}

Let us describe the interpretation of $\Path(A,a,a')$.
We define $V_\Path$ as the following pullback:
\[ \xymatrix{ V_\Path \ar[rr] \ar[d] \pb                                        & & \hU \times \hU \ar[d]^{p_A \times p_A} \\
              \U^\I \ar[rr]_-{\langle \ev \circ \lcon, \ev \circ \rcon \rangle} & & \U \times \U
            } \]
Let $E_\Path = \hU^\I$ and $p_\Path = \langle \pi \circ -, \langle \ev \circ \lcon, \ev \circ \rcon \rangle \rangle : E_\Path \to V_\Path$.
We assume that $p_\Path$ is classified by a map $\chi_\Path : V_\Path \to \U$.
The maps $v_A : \Gamma \times \I \to \U$ and $a, a' : \Gamma \to \hU$ determine a map $v_\Path' : \Gamma \to V_\Path$.
We define the interpretation of $\Path(A,a,a')$ as $\chi_\Path \circ v_\Path'$.

If $a : \Gamma \times E_\Path$ is a section of $\pi$ over $v_A$, then we define $\pcon(a) : \Gamma \to \hU^\I$ as the map corresponding to $a$ via the adjunction.
If $p : \Gamma \to E_\Path$ is a section of $p_\Path$ over $v_\Path'$ and $i : \Gamma \to \I$, then we can define the interpretation of $@$ as $\Gamma \xrightarrow{\langle p, i \rangle} \hU^\I \times \I \xrightarrow{\ev} \hU$.
A straightforward computation shows that all computation rules hold for this interpretation.

The identity type $a =_A a'$ can be defined as $\Path(x. A, a, a')$.
Its constructor $\fs{refl}(a)$ is defined as $\pcon(x. a)$.
The J rule also can be defined \cite[Section~3.1]{alg-models}.
The only problem is that J satisfies its computational rule only propositionally.
To fix this problem, we can add another computational rule for $\coe$:
\[ \coe(x. A, a, i) \equiv a \text{ if } x \notin \fs{FV}(A) \]

To show that this rule can be interpreted in our model, we define two maps $f$ and $g$ as described in the previous section.
Let $f$ be the map $\langle \fs{const} \circ \pi, \id \rangle : \hU \to T$, where $\fs{const} : \U \to \U^\I$ is the map corresponding to the projection via the adjunction.
The map $f$ is a cofibration since $e \circ f = \id$ and monomorphisms are cofibrations.
Let $g : \hU \times \I \to \hU$ be the first projection.
It is easy to see that $f$ and $g$ satisfy the required conditions.
Now, if $x \notin A$, then $v_A : \Gamma \times \I \to \U$ factors through $g : \hU \times \I \to \hU$.
It follows that the map $c : \Gamma \to T$ defined in the previous section factors through $f : \hU \to T$, which gives us the required computational rule.

\section{Univalence}

The univalence is defined as follows:

\medskip
\begin{center}
\Axiom$\fCenter \Gamma \vdash A$
\noLine
\UnaryInf$\fCenter \Gamma \vdash B$
\def\extraVskip{1pt}
\Axiom$\fCenter \Gamma, x : A \vdash b : B$
\noLine
\UnaryInf$\fCenter \Gamma, y : B \vdash a : A$
\Axiom$\fCenter \Gamma, x : A \vdash p : a[y := b] = a$
\noLine
\UnaryInf$\fCenter \Gamma, y : B \vdash q : b[x := a] = y$
\def\extraVskip{2pt}
\AxiomC{$\Gamma \vdash i : I$}
\QuaternaryInfC{$\Gamma \vdash \iso(A, B, x.b, y.a, x.p, y.q, i)$}
\DisplayProof
\end{center}

\begin{align*}
& \iso(A, B, x.b, y.a, x.p, y.q, \lcon) \equiv A \\
& \iso(A, B, x.b, y.a, x.p, y.q, \rcon) \equiv B \\
& \coe(i. \iso(A, B, x.b, y.a, x.p, y.q, i), a_0, \rcon) \equiv b[x := a_0] \text { if } i \notin \fs{FV}(A\ B\ b\ a\ p\ q)
\end{align*}

If $q : \U' \to \U$ is a trivial fibration, then $q^*(\pi)$ is a universe that classifies the same class of fibrations as $\pi$ since all objects are cofibrant.
Thus, the constructions in the previous sections apply to $q^*(\pi)$.
We cannot prove that $\iso$ can be interpreted in any universe $\pi : \hU \to \U$, but we will show that, for every $\U$, there is a trivial fibration $q : \U' \to \U$ such that $q^*(\pi)$ interprets $\iso$.

Let $\Eq(\U)$ be the object over $\U \times \U$ of equivalences between these two types (it can be defined as the object of bi-invertible maps).
If $\pi' : \hU' \to \U'$ is a pullback of $\pi$ along some trivial fibration $q : \U' \to \U$, then $\Eq(\U')$ fits in the following pullback square:
\[ \xymatrix{ \Eq(\U') \ar@{->>}[r]^{q_2'} \ar@{->>}[d]_{\langle q_0', q_1' \rangle} \pb    & \Eq(\U) \ar@{->>}[d]^{\langle q_0, q_1 \rangle} \\
              \U' \times \U' \ar@{->>}[r]_-{q \times q}                                     & \U \times \U
            } \]
Let $i : \U \to \Eq(\U)$ be the trivial cofibration corresponding to the trivial equivalence (we can actually take any section of $q_0$).
We note that there is a trivial cofibration $i' : \U' \to \Eq(\U')$ defined as $\langle \langle \id, \id \rangle, i \circ q \rangle$.
This map is a cofibration since it is a section.
It is a weak equivalence by the 2-out-of-3 property since $q_2' \circ i' = i \circ q$ and $q_2'$, $i$, and $q$ are weak equivalences.

To define the interpretation of $\iso$ in $\U'$, we need to find a cofibration $\Eq(\U') \to \U'^\I$ over $\U' \times \U'$.
Since $\U$ is a univalent universe, $q_0 : \Eq(\U) \to \U$ is a trivial fibration.
Thus, we can take $\U' = \Eq(\U)$, $q = q_0$, and $\pi' = q_0^*(\pi)$.
Then $\langle q_0', q_1', q_2' \rangle : \Eq(\U') \to \U' \times \U' \times \U'$ is a monomorphism.
Since there is a cofibration $1 \amalg 1 \amalg 1 \to \I$, we have a fibration $\U'^\I \to \U' \times \U' \times \U'$.
Now, consider the following diagram:
\[ \xymatrix{ \U \ar[rr]^-{d \circ i} \ar[d]_{i' \circ i}                  & & \U'^\I \ar@{->>}[d] \\
              \Eq(\U') \ar[rr]_-{\langle q_0', q_1', q_2' \rangle} \ar@{-->}[urr]   & & \U' \times \U' \times \U'
            } \]
where $d : \U' \to \U'^\I$ is the constant map.
Since $i' \circ i$ is a trivial cofibration, we have a lift in this diagram and since it is a left factor of a monomorphism, it is also a monomorphism.

The problem with this construction is that we cannot interpret it together with the rule for $\coe$ defined in the previous section.
We need the image of the map $\Eq(\U') \to \U'^\I$ to interact well with the image of $d : \U' \to \U'^\I$.
To be more precise, we define $\coe$ on the image of the first map as the application of the function corresponding to the equivalence, but it is defined as the identity function on the image of $d$.
Thus, the intersection of this subobjects should be contained in $i' : \U' \to \Eq(\U')$ as a subobject of $\Eq(\U')$.
In this case, two interpretations of $\coe$ will agree on the intersection.

The problem is that we do not have control over the intersection of $\Eq(\U') \to \U'^\I$ and $d : \U' \to \U'^\I$.
Note that we cannot take $i'$ instead of $i' \circ i$ and $d$ instead of $d \circ i$ in the square above because it will not commute.
To fix this problem, we construct a lift in another commutative square.
First, consider the following diagram:
\[ \xymatrix{ \U \ar[rrr]^-{d \circ i} \ar[d]_{i}                                                   & & & \U'^\I \ar@{->>}[d] \\
              \U' \ar[rrr]_-{\langle \id, \id, i \circ q, i \circ q \rangle} \ar@{-->}[urrr]^{d'}   & & & \U' \times \U' \times \U' \times \U'
            } \]
Since $4 \to \I$ is a cofibration, we have a fibration on the right.
Since $i$ is a trivial cofibration, we have a lift $d'$.

Now, consider the following diagram:
\[ \xymatrix{ \U \ar[rrr]^i \ar[d]_i                                                                        & & & \U' \ar[d]^d \\
              \U' \ar[rrr]^{d'} \ar[d]_{i'}                                                                 & & & \U'^\I \ar@{->>}[d] \\
              \Eq(\U') \ar[rrr]_-{\langle q_0', q_1', q_2', i \circ q \circ q_0' \rangle} \ar@{-->}[urrr]^h & & & \U' \times \U' \times \U' \times \U'
            } \]
We have a lift $h$ which is a monomorphism as before.
We claim that it has the required property:

\begin{lem}[intersection]
If $X$ is the intersection of $h$ and $d$ with inclusions $c_1 : X \to \U'$ and $c_2 : X \to \Eq(\U')$, then $i' \circ c_1 = c_2$.
\end{lem}
\begin{proof}
First, let us prove that the outer rectangle in the diagram above is a pullback.
Let $c_1 : X \to \U'$ and $c_2 : X \to \Eq(\U')$ be maps such that the obvious square commutes (that is, such that $q_0' \circ c_2 = q_1' \circ c_2 = q_2' \circ c_2 = i \circ q \circ q_0' \circ c_2 = c_1$).
Since $i$ is a monomorphism, we just need to find a map $t : X \to \U$ such that $i \circ t = c_1$ and $i' \circ i \circ t = c_2$.
Let $t = q \circ q_0' \circ c_2$.
We have $i \circ t = i \circ q \circ q_0' \circ c_2 = c_1$ by assumption.
Since $\langle q_0', q_1', q_2' \rangle$ is a monomorphism, to show that $i' \circ i \circ q \circ q_0' \circ c_2 = c_2$, it is enough to show that these maps become equal after we compose them with $q_0'$, $q_1'$, and $q_2'$.
Note that we have $q_j' \circ i' \circ i = i$ for every $j \in \{ 0, 1, 2 \}$.
Thus, $q_j' \circ i' \circ i \circ q \circ q_0' \circ c_2 = i \circ q \circ q_0' \circ c_2 = c_1 = q_j' \circ c_2$.

Now, if $X$ is the intersection of $h$ and $d$, then the above properties hold for it.
In particular, $i' \circ c_1 = i' \circ i \circ q \circ q_0' \circ c_2 = c_2$.
\end{proof}

Now, we can describe the interpretation of $\iso$ in $\mathcal{M}$ with $\pi' = q_0^*(\pi) : \hU' \to \U'$ as the universe (where $\U' = \Eq(\U)$ and $q_0 : \U' \to \U$ is defined as before).
The interpretation of the first six judgements in the premise of $\iso$ can be encoded as a map $\Gamma \to \qEq(\pi')$, where $\qEq(\pi')$ is the type of quasi-equivalences.
To define the interpretation of $\iso$, it is enough to define a map $h' : \qEq(\pi') \to \U'^\I$.
The first two computational rule hold if $h'$ is a map over $\U' \times \U'$.
Since $\qEq(\pi')$ is equivalent to $\Eq(\U')$ over $\U' \times \U'$ (and actually over the object of maps), it is enough to define a map $h : \Eq(\U') \to \U'^\I$ and we already did that.

To make the third computational rule hold, we need to use the construction from section~\ref{sec:interval}.
Let $C = \Eq(\U') \times_{\U'} \times \hU'$.
Let $f : C \to T$ be the following map:
\[ h \times_{\U'} \id : \Eq(\U') \times_{\U'} \times \hU' \to \U'^\I \times_{\U'} \times \hU'. \]
Now, consider the following diagram:
\[ \xymatrix{ \hU' \times \I \amalg_{\hU' \times 2} C \times 2 \ar[rr]^-{[\pi_1,[\pi_2,\ev]]} \ar[d]_{\langle i' \circ \pi', \id \rangle \square [\lcon,\rcon]} & & \hU' \ar[d]^{\pi'} \\
              C \times \I \ar[rr] \ar@{-->}[urr]^g                                                                                                              & & \U'
            } \]
where the bottom map is the map that appears in the diagram from section~\ref{sec:interval}.
Since $\langle i' \circ \pi', \id \rangle : \hU' \to C$ is a trivial cofibration, we have a lift $g$ in the diagram above.
The required properties for $f$ and $g$ follows from commutativity of this diagram.
The map $\ev : C \to \hU'$ evaluates the equivalence on the given value.
This is precisely the right hand side of the last computational rule for $\iso$, which implies that it holds for this interpretation.

Finally, we can show that this interpretation is consistent with the interpretation from the previous section.
To do this, we need to consider the intersection $Y$ of $C$ and $\hU'$ in $T$.
It can be described as the following pullback:
\[ \xymatrix{ Y \ar[r]^{d_1} \ar[d]_{d_2} \pb   & \hU' \ar[d]^{\pi'} \\
              X \ar[r]^{c_1} \ar[d]_{c_2} \pb   & \U' \ar[d]^{\ev \circ \lcon} \\
              \Eq(\U') \ar[r]_-h                & \U'^\I
            } \]
The inclusion $Y \to C$ is defined as $\langle c_2 \circ d_2, d_1 \rangle$.
Now, we need to show that the following square commutes:
\[ \xymatrix{ Y \times \I \ar[rr]^{d_1 \times \id} \ar[d]_-{\langle c_2 \circ d_2, d_1 \rangle \times \id}  & & \hU' \times \I \ar[d]^{\pi_1} \\
              C \times \I \ar[rr]_-g                                                                        & & \hU'
            } \]
By \rlem{intersection}, we have $c_2 \circ d_2 = i' \circ c_1 \circ d_2 = i' \circ \pi' \circ d_1$.
Thus, $\langle c_2 \circ d_2, d_1 \rangle = \langle i' \circ \pi', \id \rangle \circ d_1$.
Now, the required property follows from the definition of $g$ since we have $g \circ (\langle i' \circ \pi', \id \rangle \times \id) = \pi_1$.

\bibliographystyle{amsplain}
\bibliography{ref}

\end{document}